\documentclass[preprint]{imsart}
\usepackage[hyperindex,breaklinks]{hyperref}
\usepackage{graphicx}
\usepackage{natbib}
\usepackage[scriptsize]{subfigure}
\usepackage{epstopdf}
\usepackage[font=footnotesize,labelfont=bf]{caption}
\usepackage{amsmath}
\usepackage{bbm,bm}
\usepackage{amsmath,amssymb}
\usepackage{amsfonts}
\usepackage{float}
\usepackage{url}
\usepackage{algorithm}
\usepackage{algpseudocode}
\algrenewcommand\alglinenumber[1]{\tiny #1:}

\usepackage{amsthm}
\numberwithin{equation}{section}
\theoremstyle{plain}
\newtheorem{theorem}{Theorem}[section]
\newtheorem{lemma}[theorem]{Lemma}
\newtheorem{proposition}[theorem]{Proposition}
\newtheorem{definition}[theorem]{Definition}

\newtheorem{remark}[theorem]{Remark}
\newtheorem{fact}[theorem]{Fact}

\newcommand{\tht}{\bm{\theta}}
\newcommand{\omg}{\bm{\omega}}

\newcommand{\F}{\mathcal{F}}

\newcommand{\E}{\mathbb{E}}
\newcommand{\PP}{\mathbb{P}}

\newcommand{\ol}{\overline}

\newcommand{\field}[1]{\mathbbm{#1}}
\newcommand{\R}{\field{R}}
\newcommand{\C}{\field{C}}

\newcommand{\N}{\field{N}}

\newcommand{\Var}{\mathrm{Var}}

\newcommand{\dtv}{d_{\mathrm{TV}}}
\newcommand{\mtx}[1]{\bm{#1}}

\begin{document}
\begin{aug}
\title{Hypothesis testing for Markov chain Monte Carlo}
\author{Benjamin M. Gyori         \and Daniel Paulin
}

\runtitle{Hypothesis testing for Markov chain Monte Carlo}
\runauthor{BM Gyori and D Paulin} 
\address{Department of Systems Biology, Harvard Medical School \\
              E-mail: ben.gyori@gmail.com\\           %  \\
           \and
           Department of Statistics and Applied Probability, National University of Singapore\\
           E-mail: paulindani@gmail.com
}

\date{}

\begin{abstract}
Testing between hypotheses, when independent sampling is possible, is a well developed subject. In this paper, we propose hypothesis tests that are applicable when the samples are obtained using Markov chain Monte Carlo. These tests are useful when one is interested in deciding whether the expected value of a certain quantity is above or below a given threshold. We show non-asymptotic error bounds and bounds on the expected number of samples for three types of tests, a fixed sample size test, a sequential test with indifference region, and a sequential test without indifference region. Our tests can lead to significant savings in sample size. We illustrate our results on an example of Bayesian parameter inference involving an ODE model of a biochemical pathway.\footnote{Source code available at github.com/bgyori/mcmchyp.}
\end{abstract}

%\vspace{10pt}

\end{aug}
\maketitle

\makeatletter{}\section{Introduction}
The goal of Markov chain Monte Carlo estimation is to calculate an expected value with respect to a probability distribution from which sampling directly is impossible or impractical. 

It is often only of interest whether the expected value is above or below a certain threshold (such as whether $\E_{\pi} f > r$, where $\pi$ is a probability distribution and $f$ is the function of interest). This problem can be posed as a decision between two hypotheses, and one must make a decision with a bounded error probability. This has applications in the verification of stochastic systems (here one aims to verify whether a model satisfies a property with at least a given probability), arising in software testing, robotics and systems biology \citep{legay2010statistical}. 

In the case of independent samples, optimal fixed length and sequential hypothesis tests are available for deciding between two hypotheses (see \cite{wald1945sequential}, \cite{laitzeleung1}, \cite{laitzeleung2}, \cite{Lehmann}). However, such tests are not available in the case of samples obtained by MCMC simulation.

Our main contribution in this paper is the introduction of one fixed length test and two sequential tests (one with indifference region, and one without indifference region) that allow us to decide whether the expected value of a quantity exceeds a certain threshold. We prove non-asymptotic bounds for the probability of error (choosing the incorrect hypothesis) and the expected running times of these tests. The advantage of our approach is that the sample size needed to make the decision between the two hypothesis can be much shorter than the one needed to precisely estimate the expected value.

We note that hypothesis testing ideas have been used recently in the context of approximate MCMC algorithms in \cite{Korattikara} and \cite{BardenetMCMC}. In these papers, subsampling is used to approximate the log-likelihood, speeding up the computation of every MCMC step, at the price of sampling from an approximate distribution instead of the true one. Hypothesis testing ideas are used to bound the distance of the resulting distribution and the target distribution in total variational distance. In this paper, we take a different approach, by reducing the amount of MCMC steps needed to decide whether $\E_{\pi} f > r$ for some function $f$. The two approaches are complimentary to each other, and could, in principle be combined.

The paper is organized as follows. In Section \ref{section_preliminary}, we briefly review the necessary preliminaries from the theory of Markov chains. In Section \ref{sec:samplesize}, we state the three hypothesis tests, and their theoretical properties. Finally, in Section \ref{sec:results}, we evaluate these tests on an ODE model of a biochemical pathway, whose parameter posterior is explored using MCMC.

\makeatletter{}\section{Markov chain preliminaries}\label{section_preliminary}
In this section, we review some basic definitions about Markov chains.
\subsection{Spectral gap of general state Markov chains}
Firstly, we state the definition of the spectral gap of reversible Markov chains following \cite{Robertsgeneral} (see \cite{Kato} for more on the spectral properties of linear operators). The spectral gap is a measure of how fast the chain is mixing, and will be needed to state our error bounds for the hypothesis tests in this paper. The main reason we are using the spectral gap is that it is the basis of existing non-asymptotic bounds for MCMC empirical averages (see \cite{leon2004optimal}) that are known to be sharp in some cases. 

We call a Markov chain $X_1,X_2,\ldots$ on state space $(\Omega,\F)$ with transition kernel $P(x,dy)$ \emph{reversible} if there exists a probability measure $\pi$ on $(\Omega,\F)$ satisfying the detailed balance conditions,
\begin{equation}
\pi(dx) P(x,dy)= \pi(dy) P(y,dx) \text{ for  every } x,y\in \Omega.
\end{equation}
Define $L_2(\pi)$ as the Hilbert space of complex valued measurable functions that are square integrable with respect to $\pi$, endowed with the inner product $\left<f,g\right>_{\pi}=\int f g^*\, \mathrm{d}\pi$. $P$ can be then viewed as a linear operator on $L_2(\pi)$, denoted by $\mtx{P}$, defined as \[(\mtx{P} f)(x):=\E_{P(x,\cdot)}(f),\] and reversibility is equivalent to the self-adjointness of $\mtx{P}$.
The operator $\mtx{P}$ acts on measures to the left, i.e.\ for every measurable subset $A$ of $\Omega$, \[(\mu \mtx{P})(A):=\int_{x\in \Omega} P(x,A) \mu(\mathrm{d} x).\]
For a Markov chain with stationary distribution $\pi$, we define the \emph{spectrum} of the chain as
\begin{align*}
S_2:=\{\lambda\in \C\setminus 0:& (\lambda\mtx{I}-\mtx{P})^{-1}\text{ does not exist as a}\\
&\text{bounded linear operator on } L_2(\pi)\}.
\end{align*}
For reversible chains, $S_2$ lies on the real line. The following is the main definition of this section.
\begin{definition}\label{spectralgapdef}
The \emph{spectral gap} for reversible chains is defined as
\begin{equation*}
\gamma :=
\begin{cases}
1-\sup\{\lambda: \lambda\in S_2, \lambda\ne 1\}  & \text{if eigenvalue 1 has} \\
& \text{multiplicity 1, and } \\
0  & \text{otherwise}.
\end{cases}
\end{equation*}
Similary, we define the \emph{absolute spectral gap} as
\begin{equation*}
\gamma^*:=
\begin{cases}
1-\sup\{|\lambda|: \lambda\in S_2, \lambda\ne 1\}  & \text{if eigenvalue 1 has} \\
& \text{multiplicity 1, and } \\
0  & \text{otherwise}.
\end{cases}
\end{equation*}

\end{definition}
It follows from the definition that $\gamma\ge \gamma^*$. They are often equal in practice, and it is simple to modify the MCMC steps such that $\gamma=\gamma^*$ (by considering the so-called ``lazy'' version of the chain). 
For reversible chains, it is known that $\gamma^*>0$ implies geometric ergodicity, and the CLT for any function $f\in L_2(\pi)$. Moreover, for any initial distribution $\nu$, and any $k$,
\[\dtv(\nu P^k, \pi)\le (1-\gamma^*)^k \left\|\frac{\mathrm{d}\nu}{\mathrm{d}\pi}-1\right\|_{2,\pi},\]
so the absolute spectral gap is related to the speed of convergence to equilibrium (in total variational distance).

In the case of non-reversible chains, \cite{Martoncoupling} defines the pseudo-spectral gap, and shows that it has similar properties as the spectral gap has for reversible chains.

\subsection{Estimating the spectral gap}
In practice, the spectral gap is often not known, and has to be estimated from the output of the chain. 
In the Appendix, we propose an estimator based on the following simple fact.
\begin{fact} Suppose that $f:\Omega\to \R$ is such that $\E_{\pi}(f)=0$, and it satisfies an additional technical assumption. Let $\rho_{\eta}(f):=\E_{\pi}(f(X_0)f(X_{\eta}))$, then
\begin{equation}\lim_{\eta\to \infty}(\rho_{\eta}(f)/\Var_{\pi}(f))^{1/\eta}=1-\gamma^*.\end{equation}
\end{fact}
From the output of the chain we can estimate $\rho_{\eta}(f)$, which in turn leads to estimates for $\gamma^*$. For further details on the exact procedure, as well as some numerical results, we refer the reader to the Appendix. One important point to note is that the proposed procedure is quite fast, and does not affect the overall running time significantly.

\makeatletter{}\section{Hypothesis tests and error bounds}\label{sec:samplesize}
In this section we present bounds on the error of the hypothesis tests on MCMC estimates. First we review a Hoeffding-type inequality for reversible Markov chains, and then introduce several hypothesis tests based on it. 

An important assumption of the concentration inequalities is that the Markov chain is stationary. To ensure this, we are going to discard the first $t_0$ terms of the Markov chain ($t_0$ is commonly called the burn-in time \citep{GilksMCMC}), and only take into account the terms $f(X_{t_0+1}),\ldots,f(X_{n})$. Similarly to \cite{nonasymptotic}, we will set $t_0 \ge 30/\gamma$. In what follows, we will not show $t_0$ explicitly, and assume that sufficient burn-in steps were performed such that $f(X_{1})$,$\ldots$,$f(X_{n})$ is approximately stationary.

\subsection{Concentration of the MCMC estimate}
In this section, we review a Hoeffding-type inequality for reversible Markov chains (this will be used to bound the errors of our hypothesis tests).
This result was proven for Markov chains on finite state spaces in \cite{leon2004optimal}, and extended to general state spaces in \cite{MiasojedowHoeffding}.
\begin{theorem}[Hoeffding inequality for reversible chains]\label{thmhoeffdingrev}
Let $X_1,\ldots,$ $X_n$ be a stationary, reversible Mar\-kov chain with spectral gap $\gamma$, and unique stationary distribution $\pi$. Let $f\in L^2(\pi)$ such that $0 \le f(x)\le 1$ for every $x\in \Omega$. Let $S_n:=\sum_{i=1}^n f(X_i)$, then for any $t\ge 0$, 
\begin{equation}\label{Hoeffdingrevempeq}
\PP(|S_n-n\cdot \E_{\pi} f|\ge t)\le 2 \exp\left(-\frac{t^2\cdot \gamma}{n}\right).
\end{equation}
\end{theorem}
This form of the result follows from equation (3) of  \cite{leon2004optimal} by rescaling. A similar inequality can be shown to hold for non-reversible chains with the spectral gap being replaced by the pseudo spectral gap, see \cite{Martoncoupling} for more details.

In the case of non-stationary chains, one can show that essentially the same result holds (see Propositions 3.3 and 3.4 of \cite{Martoncoupling}).

\subsection{Hypothesis tests with indifference region}
Suppose that $X_1,X_2,\ldots$ is a reversible Markov chain taking values in a Polish state space $\Omega$, with unique stationary distribution $\pi$, and $f:\Omega \to [0,1]$ is a bounded function (with a simple scaling argument, our results extend to the case when $f:\Omega \to [\alpha,\beta]$). 

Our first objective is to do a test between the following two hypotheses, given $r\in (0,1)$ and $\delta\in (0,\min(r,1-r))$.
\begin{align}\label{eq:hypindiff}
H_0&: \E_{\pi} f\ge r+\delta, \\
H_1&: \E_{\pi} f\le r-\delta. \nonumber
\end{align}
Here $(r-\delta,r+\delta)$ is an indifference region in which choosing either hypothesis is acceptable. 

We first discuss two tests to choose between these hypotheses. The first one is a fixed sample size test, while the second one is a sequential test. 

\subsubsection{Fixed length hypothesis test}
Suppose that we have a sample of length $n$ consisting of the values $f(X_1),\ldots$, $f(X_n)$. The fixed length hypothesis test is shown in Algorithm \ref{alg:fixed}.

\begin{algorithm}[h]
\linespread{1}
\renewcommand{\thealgorithm}{1}
\caption{Fixed length hypothesis test}
\small
\label{alg:fixed}
Input: Threshold $r$, number of samples $n$

Output: Choice of $H_0$ or $H_1$.
	\begin{algorithmic}
		\State Set $S_0 := 0$
		\For{$i:=1 \ldots n$}
			\State $S_i := S_{i-1} + f(X_i)$
		\EndFor	
		
		\If{$S_n \ge nr$}
			\State \textbf{return} $H_0$
		\Else
			\State \textbf{return} $H_1$
		\EndIf
	\end{algorithmic}
\end{algorithm}

The next proposition bounds the error probability of the test in Algorithm \ref{alg:fixed}, where an error constitutes accepting hypothesis $H_1$ when in fact hypothesis $H_0$ holds, and vice-versa.
\begin{proposition}[Error bound for fixed length hypothesis test]\label{propdettest}
For the fixed length hypothesis test, the error rate is bounded by
\begin{equation}
\exp\left(-\gamma \delta^2 n\right).
\end{equation}
\end{proposition}
\begin{proof}
Suppose that $H_1$ holds, implying that $\E_{\pi}f \le r-\delta$. An error is made ($H_0$ is chosen) if $S_n \ge nr$. From here 
\begin{equation*}
S_n-n\E_{\pi}f \ge nr-n(r-\delta) = n\delta.
\end{equation*}
Applying the Hoeffding inequality (Theorem \ref{thmhoeffdingrev}), we get
\[\PP(S_n-n\E_{\pi}f \ge n\delta) \le \exp\left( -\gamma\delta^2 n \right).\]
The same holds under the opposite hypothesis.
\end{proof}

This implies, in particular, that if we want the error to be smaller than $\epsilon$, then
\begin{equation}
\label{eq:nbound}n\ge \frac{\log(1/\epsilon)}{\gamma \delta^2}
\end{equation}
samples are sufficient.

\subsubsection{Sequential hypothesis test}
The main idea behind a sequential test is to monitor the empirical sum, and stop if the samples collected so far are sufficient to decide between the hypotheses. Such a sequential test is shown in Algorithm \ref{alg:seq}. Although testing in every step is intuitive, because we bound the error by a union bound based on concentration inequalities, we obtain sharper results if we test not in every step, but only at stages $n_i:=\left\lfloor n_0 (1+\xi)^i\right\rfloor$ for some parameters $\xi>0$, and $n_0$. We will choose $n_0:=\left\lfloor M\min\left(\frac{1}{1-r},\frac{1}{r}\right)\right\rfloor$.

\begin{algorithm}[h]
\caption{Sequential hypothesis test with indifference region}
\small
\label{alg:seq}
Input: Threshold $r$, stopping condition $M$, testing param. $\xi$

Output: Choice of $H_0$ or $H_1$.
	\begin{algorithmic}
		\State Set $S_{n_{0}} := \sum_{1\le k\le n_0}f(X_k)$ and $i:=1$
		\Loop
			\State $S_{n_i} := S_{n_{i-1}} + \sum_{n_{i-1}<k\le n_i}f(X_k)$
			\If{$S_{n_i} \ge n_i r + M$}
				\State \textbf{return} $H_0$
			\ElsIf{$S_{n_i} \le n_i r - M$}
				\State \textbf{return} $H_1$
			\Else
				\State Set $i:=i+1$ and continue
			\EndIf
		\EndLoop	
	\end{algorithmic}
\end{algorithm}

The following proposition bounds the error probability of the test in Algorithm \ref{alg:seq} with a particular choice of $M$.
\begin{proposition}[Error bound for sequential hypothesis test]\label{propseqtest}
Suppose that $\xi\le 0.4$ and $\epsilon\le 0.4$, and choose
\begin{equation}\label{eq:seqtestM}
M:= \frac{\log\left(2/\sqrt{\epsilon \xi}\right)}{2\gamma\delta}.
\end{equation}
For the sequential hypothesis test, the probability of an error is bounded by $\epsilon$.
\end{proposition}
\begin{proof} 
Let $\epsilon':=\frac{\sqrt{\epsilon \xi}}{2}$, then $M=\frac{\log\left(1/\epsilon'\right)}{2\gamma\delta}$. Suppose $H_1$ holds, implying that $\E_{\pi}f \le r-\delta$.
It is easy to see that the probability of choosing $H_0$ is bounded by the sum
\begin{equation}\label{type1sumeq}
\sum_{i=1}^{\infty}\PP(S_{n_i}\ge n_i r+M).
\end{equation}
By Theorem \ref{thmhoeffdingrev}, for any $i\ge 1$, we have
\begin{align*}
&\PP(S_{n_i}\ge n_{i}r+M)\le \PP(S_{n_i}\ge n_i\E_{\pi}f +n_i\delta+M)\\
&\le \exp\left(-\frac{\gamma (n_i\delta+M)^2}{n_i}\right)\\
&= \exp\left(-\gamma \left(n_i\delta^2 +2M\delta  +\frac{M^2}{n_i}\right)\right)\\
&=\exp\left(-\gamma n_i \delta^2 - \log\left(1/\epsilon'\right)-\frac{\log^2\left(1/\epsilon'\right)}{4 \gamma n_i \delta^2}\right)\\
&=\exp\left(-\log\left(1/\epsilon'\right)\left(1+\frac{1}{2} \frac{2\gamma n_i \delta^2}{\log\left(1/\epsilon'\right)}+\frac{1}{2} \frac{\log\left(1/\epsilon'\right)}{2\gamma n_i \delta^2}\right)\right).\end{align*}
Let $x_i:=\frac{2\gamma n_i \delta^2}{\log\left(1/\epsilon'\right)}$, 
then we have
\begin{align*}
&\sum_{i=1}^{\infty}\PP(S_{n_i}\ge n_{i}r+M) \\
&\le \sum_{i\ge 1}\exp\left(\log(\epsilon')\left(1+\frac{1}{2}(x_i+1/x_i)\right)\right)\\ 
&=\epsilon'\sum_{i\ge 1}\exp\left(\frac{\log(\epsilon')}{2}(x_i+1/x_i)\right).
\end{align*}
Note that $x+1/x$ is monotone decreasing in the interval $(0,1]$, and monotone increasing in the interval $[1,\infty]$.
Now using the definitions $n_i:= \lfloor n_0(1+\xi)^i\rfloor$, it is easy to see that 
we can upper bound this sum by replacing the sequence of $x_i$s with a sequence of $(1+\xi)^k$ for $k=0,1,\ldots$, and multiplying that sum by two, that is,
\begin{align*}
&\sum_{i=1}^{\infty}\PP(S_{n_i} \ge n_{i}r+M)\\
&\le 2\epsilon' \sum_{k=0}^{\infty}
\exp\left(\frac{\log(\epsilon')}{2}\left((1+\xi)^k+(1+\xi)^{-k} \right)\right).
\end{align*}
Using the assumption $0<\xi\le 0.4$, one can show that $(1+\xi)^k+(1+\xi)^{-k}\ge 2(j+1)$ for $\frac{ 1.6\cdot j}{\xi}\le k< \frac{ 1.6\cdot (j+1) }{\xi}$, with $j\in \N$. Therefore the above sum can be bounded as
\begin{align*}
&\sum_{i=1}^{\infty}\PP(S_{n_i} \ge n_{i}r+M)\\
&\le  2\epsilon' \cdot \frac{1.6}{\xi}\cdot \sum_{j=0}^{\infty}(\epsilon')^{j+1}\le \frac{3.2(\epsilon')^2}{\xi (1-\epsilon')}\le \frac{4(\epsilon')^2}{\xi}\le \epsilon,
\end{align*}
where we have used the assumptions that  $\xi\le 0.4$ and $\epsilon\le 0.4$, implying that $\epsilon'\le 1/5$.
The same holds under the opposite hypothesis.
\end{proof}

\subsection{Sequential hypothesis test without indifference region}
Suppose that, unlike previously, we do not want to use an indifference region.
This case arises when our objective is to do a test between the following two hypotheses given $r\in (0,1)$:
\begin{align*}
H_0&: \E_{\pi} f> r, \\
H_1&: \E_{\pi} f< r.
\end{align*}

In this case, we cannot use a fixed length hypothesis test. We propose the following modified version of the sequential hypothesis test of the previous section. 
Let $n_0:=\lfloor 100/\gamma \rfloor$, 
and $n_i:=\lfloor n_0 (1+\xi)^i \rfloor$ for some $\xi>0$ (we have chosen $n_0$ this way because we will need to run the chain at least this long for estimating the spectral gap). Let $\epsilon>0$ be the specified error probability of the test, and for $i\ge 1$ let 
\begin{align}\label{eq:gneps}
g(i,\epsilon):=\left(\frac{n_i}{\gamma}\cdot \left(\log(1/\epsilon)+1+2\log(i) \right)\right)^{1/2}.
\end{align}

\begin{algorithm}[h]
\caption{Sequential hypothesis test without indifference region}
\small
\label{alg:seqnoindiff}
Input: Threshold $r$, error bound $\epsilon$, testing param. $\xi$

Output: Choice of $H_0$ or $H_1$.
	\begin{algorithmic}
		\State Set $S_{n_0} := \sum_{1\le k\le n_0} f(X_k)$ and $i:=1$
		\Loop
			\State $S_{n_i} := S_{n_{i-1}} + \sum_{n_{i-1}<k\le n_i}f(X_k)$
			\If{$S_{n_i} \ge n_i r + g(i,\epsilon)$}
				\State \textbf{return} $H_0$
			\ElsIf{$S_{n_i} \le n_i r - g(i,\epsilon)$}
				\State \textbf{return} $H_1$
			\Else
				\State Set $i:=i+1$ and continue
			\EndIf
		\EndLoop	
	\end{algorithmic}
\end{algorithm}

The following proposition bounds the error probability of the test in Algorithm \ref{alg:seqnoindiff}.
\begin{proposition}[Error bound for sequential hypothesis test without indifference region]\label{propseqnoindefftest}
For the test explained in Algorithm \ref{alg:seqnoindiff}, the error probability is bounded by $\epsilon$.
\end{proposition}
\begin{proof}[Proof of Proposition \ref{propseqnoindefftest}]
Suppose $H_1$ holds, implying that $\E_{\pi}f < r$.
Using Hoeffding's inequality for reversible Markov chains (Theorem \ref{thmhoeffdingrev}), the probability of choosing $H_0$ is bounded by
\begin{align*}
&\sum_{i=1}^{\infty} \PP(S_{n_{i}}\ge n_ir+g(i,\epsilon))\\
&\le \sum_{i=1}^{\infty} \PP(S_{n_i}-n_i\E_{\pi}f \ge g(i,\epsilon)) \\
&\le \sum_{i=1}^{\infty} \exp(-g(i,\epsilon)^2 \cdot \gamma/n_i)\le
\frac{\epsilon}{\exp(1)}\sum_{i=1}^{\infty} \frac{1}{i^2} \le \epsilon,
\end{align*}
where $g(i,\epsilon)$ is set according to \eqref{eq:gneps}. The same holds under the opposite hypothesis.
\end{proof}

\subsection{Expected stopping times}\label{sec:stoptime}
In this section we analyze the (expected) number of samples taken in each test, given $\epsilon,r$ and $\delta$.

The fixed length hypothesis test will always take $\log(1/\epsilon)/(\gamma \delta^2)$ steps to decide between the hypotheses in \eqref{eq:hypindiff} with error at most $\epsilon$. 
We now show the \emph{expected} stopping time of the sequential test with indifference region in Proposition \ref{thm:seqet}. For conciseness, we will use the notation
\begin{equation}\label{eqDeltadef}
\Delta:=|r-\E_{\pi}(f)|,
\end{equation}
with which $\Delta \ge \delta$ holds under either hypothesis. 

\begin{proposition}\label{thm:seqet}
	For the sequential test with indifference region, with $M$ chosen according to \eqref{eq:seqtestM} as 
 $M:= \frac{\log\left(2/\sqrt{\epsilon \xi}\right)}{2\gamma\delta}$, the expected stopping time satisfies
	\begin{equation}\label{eq:ETbound}
	\E(T)\le (1+\xi)\left(\frac{M}{\Delta} + 2\sqrt{\frac{M+2\Delta}{\gamma\Delta^3} + \frac{2}{\gamma\Delta^2}}\right)
	\end{equation}
	under both hypotheses. 
\end{proposition}
\begin{proof}
Under hypothesis $H_1$, by the definition of the test, using the Hoeffding inequality we can see that
\begin{align*}
&\E(T)\le n_1+\sum_{i=2}^{\infty}(n_i-n_{i-1})\PP(T>n_{i-1})\\
&\le n_1+\sum_{i=2}^{\infty}(n_i-n_{i-1})\PP(S_{n_{i-1}}>n_{i-1} r - M)\\
& =n_1+\sum_{i=2}^{\infty}(n_i-n_{i-1})\PP(S_{n_{i-1}}-n_{i-1}\E_{\pi} f> n_{i-1} \Delta - M)\\
&\le n_1+\sum_{i=2}^{\infty}(n_i-n_{i-1})\exp\left(-\frac{\gamma (n_{i-1} \Delta - M)_{+}^2}{n_{i-1}}\right)\\
&\le n_1+(1+\xi) \int_{t=n_1}^{\infty}\exp\left(-\frac{\gamma (t \Delta - M)_{+}^2}{t}\right)\mathrm{d}t\\
&\le (1+\xi)\left(\frac{M}{\Delta}+2\sqrt{\frac{M+2\Delta}{\gamma\Delta^3} + \frac{2}{\gamma\Delta^2}} \right),
\end{align*}
where the last step follows from an upper bound for the exponential integral, and $x_+$ denotes the positive part of $x\in\R$.%
\end{proof}
The expected stopping time of the test grows essentially linearly in $M$. One can then show that under any of the two hypotheses, as $\epsilon\to 0$, 
\[\E (T)\le (1+\xi)\frac{\log\left(2/\sqrt{\epsilon \xi}\right)}{2\gamma\delta\Delta}+\mathcal{O}(\sqrt{M/\Delta}),\] 
which can be much smaller than $\log(1/\epsilon)/(\gamma\delta^2)$, the number of steps for the fixed length test, if $\Delta$ is much larger than $\delta$.  
Optimization of the above  expected stopping time bound in $\xi$ yields that the choice 
\begin{equation}\label{eq:xichoice}
\xi:=\frac{1}{\log(2)\log(1/\epsilon)}
\end{equation}
is reasonable, which gives $\xi\approx 0.3$ for $\epsilon=0.01$.

Simple arguments show that the test will stop in a finite amount of time almost surely, even if none of the two hypotheses is satisfied, that is, if $\E_{\pi}f \in (r-\delta,r+\delta)$. In practice, however, one may need to stop the run after a certain number of steps.
Note that it is easy to show that for any $t\ge 0$, $T$ satisfies the inequality
	\begin{equation}\label{eq:Tconceq}
	\PP\left(T\ge t\right) \le (1+\xi)\exp\left(-\frac{\gamma (t \Delta - M)_{+}^2}{t}\right).
	\end{equation}
Using this inequality, the probability that the chain runs for more than $6M/\delta$ steps is less than
\begin{equation}\label{eqstopseq}
(1+\xi)\exp\left(-\gamma \delta\cdot 4M\right)\le \frac{(1+\xi)\epsilon \xi}{4},
\end{equation}
under both hypotheses, which is quite small. Therefore, if this happens, we will choose $H_0$  if $S_{6M/\delta} \ge (6M/\delta)r$ and $H_1$ otherwise. This modification of the original test only changes the error at most by the amount \eqref{eqstopseq}. 

Now we turn to the sequential hypothesis test without indifference region. The following proposition bounds the expected stopping time of the test.

\begin{proposition}
	For the sequential test without indifference region, the expected stopping time satisfies
	\begin{equation}\label{eq:ETbound2}
	\E(T)\le (1+\xi)\left(N + \frac{4\epsilon}{\gamma \Delta^2}\right),
	\end{equation}
	under both hypotheses, where $\Delta$ is defined as in \eqref{eqDeltadef}, and $N$ is defined as
	\begin{align*}
	&N:=\inf\Big\{ n_i, i\ge 1: \frac{4 \left(\log(1/\epsilon)+1+2\log(i)\right)}{\gamma \Delta^2}\le n_i\Big\}.
	\end{align*}
\end{proposition}
\begin{proof}
The probability that the test takes at least $n_i$ steps can be bounded as
\begin{align*}\label{eq:Tgn}
&\PP(T\ge n_i)\le \PP(|S_{n_i} - n_ir| \le g(i,\epsilon)) \\
&=\PP(-g(i,\epsilon) \le S_{n_i} - n_ir \le g(i,\epsilon))\\
&\begin{cases}
\le \PP(S_{n_i}-n_i\E_{\pi}f \ge n_i(r-\E_{\pi}f)-g(i,\epsilon))  \text{ if } \E_{\pi}f < r\\
\le \PP(S_{n_i}-n_i\E_{\pi}f \le n_i(r-\E_{\pi}f)+g(i,\epsilon))  \text{ if } \E_{\pi}f > r,
\end{cases}
\end{align*}
and by applying the Hoeffding inequality, this can be further bounded by
\begin{equation}
\PP(T\ge n_i) \le\exp\left( -\frac{\gamma}{n_i} \left((n_i \Delta -g(i,\epsilon)\right)_+)^2\right).
\end{equation} 
By the definition of $N$, it is easy to see that for $n_i\ge N$, we have $g(i,\epsilon)\le n_i\Delta/2$. Using this, and the fact that 
\[\E(T)=n_1+\sum_{i=1}^{\infty}(n_{i+1}-n_i)\PP(T> n_i),\] 
we can show that
\begin{align*}\E(T)&\le 
(1+\xi)\left(N+\sum_{k=N+1}^{\infty}\exp\left( -\frac{\gamma}{k} \left(k \Delta/2\right)^2\right)\right)\\
&\le  (1+\xi)\left(N+ \frac{\exp\left( -\gamma(N+1)\Delta^2/4\right)}{1-\exp\left( -\gamma\Delta^2/4\right)}\right)\\
&\le (1+\xi)\left(N+  \frac{\epsilon\cdot \exp\left( -\gamma\Delta^2/4\right)}{1-\exp\left( -\gamma\Delta^2/4\right)}\right)\\
&\le (1+\xi)\left(N+\frac{4\epsilon}{\gamma\Delta^2}\right),\end{align*}
using the definition of $N$ and the fact that for $x\ge 0$, $\frac{\exp(-x)}{1-\exp(-x)}\le \frac{1}{x}$.
\end{proof} 
Our bound on the expected stopping time of the test without indifference region is similar to  sequential test with indifference region. As $\epsilon\to 0$, and $\xi\to 0$, the bounds on the expected running time of the two tests with/without indifference region are approximately 
\[\frac{\log(1/(\epsilon\xi))}{4\gamma \delta \Delta}(1+o(1)) \text{ and }\frac{4\log(1/(\epsilon\xi))}{\gamma \Delta^2}(1+o(1)),\] respectively. This means that the second sequential test performs better when $\Delta$ is considerably bigger than $\delta$, but when they are close, the first test performs better. Finally, the second test does not assume the existence of an indifference region, therefore it is more generally applicable.

\makeatletter{}\section{Case study}\label{sec:results}
Here we present a case study to evaluate various aspects of our hypothesis tests empirically. 
The case study is from the domain of systems biology. Dynamical system models based on ordinary differential equations (ODEs) are often used to describe the concentration of molecular species such as proteins inside the cell~\citep{klipp2005systems}. However, the rate constants associated with biological processes are often unknown and not directly measurable. This appears as a set of unknown parameters in the ODE model whose value can only be inferred given limited and noisy experimental data. In a Bayesian inference setting, sampling directly from the posterior distribution of parameters will not be possible, and a Markov chain Monte Carlo method can be used to collect samples from the posterior \citep{learningsysbbook}.

We assume a dynamical system model of the form
\begin{align}
\dot{\mathbf{x}}(t) &= F(\mathbf{x}(t),\tht)\\
\mathbf{y}(t) &= G(\mathbf{x}(t),\tht) + \omg(t). \nonumber
\end{align}
Here $\mathbf{x} \in \R^{d_x}$ is a vector of state variables, $\mathbf{y} \in \R^{d_y}$ is a vector of observables, $\tht \in \R^{d_{\theta}}$ is a vector of parameters, and $\omg(t) \in \R^{d_y}$ is a random variable representing measurement noise. 

The goal is to construct the posterior distribution of $\tht$ given its prior $p_0(\tht)$, and a set of observations $Y$. Here $Y$ consists of observations of the form $Y_{i,j}$,  where $1 \le i \le d_y$ denotes the $i$th observable of the system at time point $t_j$, $1 \le j \le T $. We denote the likelihood of the set of observations with respect to $\tht$ as $p(Y|\tht)$. The posterior, denoted $\pi(\tht|Y)$, is then expressed in the standard way as
\begin{equation}
\pi(\tht|Y) = \frac{p_0(\tht)p(Y|\tht)}{\int p_0(\tht)p(Y|\tht) d\tht} \propto p_0(\tht)p(Y|\tht).
\end{equation}
In this case study we use a uniform prior over a bounded set of parameters, and Gaussian likelihood (corresponding to uncorrelated multivariate Gaussian $\omg(t)$) defined as
\begin{align}\label{eq:guasslh}
p(Y|\tht) =& \prod_{i=1}^{d_y}\prod_{j=1}^{T} P(Y_{i,j}|\tht) \\
\propto & \exp \left(-\sum_{i=1}^{d_y}\sum_{j=1}^{T} \left(\frac{Y_{i,j}-y_{i}(t_j)|_{\tht}}{\sqrt{2}\sigma_{i,j}}\right)^2\right). \nonumber
\end{align}
Here $y_{i}(t_j)|_{\tht}$ denotes the $i$th observable of the model at time $t_j$, when simulated with parameters $\tht$, and $\sigma_{i,j}$ is the standard deviation associated with data point $Y_{i,j}$.

The analysis of such a system involves estimating the expected value of a function $f$ of $\tht$ with respect to $\pi(\tht|Y)$, denoted $\E_{\pi}f$.
Closed form solutions to this problem will, in general, not be available, and sampling independently from the posterior $\pi(\tht|Y)$ will not be possible. We therefore resort to using a Metropolis-Hastings chain to collect a sequence of parameters $\tht_1$, $\ldots$, $\tht_n$, and use the approximation
\begin{equation}
\E_{\pi}f \approx \frac{1}{n}\sum_{i=1}^nf(\tht_i).
\end{equation}
Each state of the chain is obtained by a proposal and an acceptance step. We use a symmetric Gaussian proposal of the form
$q(\tht_i \to \tht') = \mathcal{N}(\tht_i,\Sigma_{\mathrm{MH}})$, a $d_{\theta}$-dimensional multivariate Gaussian with mean identical to the current parameter vector, and covariance matrix $\Sigma_{\mathrm{MH}}$. Here $\Sigma_{\mathrm{MH}}$ is diagonal with entries $\sigma_{\mathrm{MH},1}^2$, $\ldots$, $\sigma_{\mathrm{MH},d_{\theta}}^2$, representing variances along each dimension independently. 
The proposed parameter $\tht'$ is accepted with probability $\alpha(\tht_i \to \tht')$, determined by the posterior as
\begin{align}
\alpha(\tht_i \to \tht') &= \min \left (1,\frac{\pi(\tht'|Y)}{\pi(\tht_i|Y)} \right) = \\
&=\min\left(1, \frac{p_0(\tht')p(Y|\tht')}{p_0(\tht_i)p(Y|\tht_i)} \right ). \nonumber
\end{align}
Note that the proposal does not appear in the acceptance ratio due to symmetry.

We apply our method to analyze the dynamics of the JAK-STAT biochemical pathway (for more details, see \cite{swameye2003}). 
Our goal will be to decide about an important property of the system, namely, whether the concentration of nuclear STAT protein reaches the threshold of $1$. 

The variables in the model represent the quantity of different forms of the STAT protein in a biological cell. These quantities cannot be measured directly, however, experimental data for two indirect quantities (total phosphorylated STAT, and total STAT in cytoplasm) has been published in \cite{swameye2003}. There are $4$ model parameters $\tht=(k_1,k_2,k_3,k_4)$ corresponding to the kinetic rate constants of biochemical reactions, whose values are unknown. The equations governing the system dynamics and the parameters used to run the MCMC chain are given in the Appendix. 

We let 
\begin{equation*}
f(\tht) := \begin{cases} 0 & \text{if nuclear STAT does not reach 1} \\ 
& \text{when simulating with } \tht, \\
1 & \text{if nuclear STAT reaches 1} \\
& \text{when simulating with } \tht.
\end{cases}
\end{equation*}

We ran $m=1000$ independent instances of the MCMC chain for a total of $2\cdot 10^6$ steps each (with an additional $t_0=5\cdot 10^4$ burn-in steps). The method described in the Appendix was used to estimate the value of the spectral gap, independently for each chain. We used the output of the $m$ independent chains as a basis for constructing empirical results in Figure \ref{fig:results_psi1}(a-f).

To get a reliable estimate of the true expected value, $\widehat{E}\approx\E_{\pi} f$, we took the overall average of the estimates from all $m$ chains, and treated the obtained value $\widehat{E} = 0.8875$ as the reference for $\E_{\pi} f$. 

We first examined the empirical error rate of the fixed sample size hypothesis test. For a fixed sample size $n$, we define the empirical error rate $E_n$ as the ratio of chains choosing $H_0$ if $H_1$ holds (or the ratio choosing $H_1$ if $H_0$ holds). If neither $H_0$ nor $H_1$ holds (when $r-\delta < \E_{\pi}f < r + \delta$), then $E_n := 0$.
We set $r=\widehat{E}-\delta$ and calculated $E_n$ for a range of sample sizes up to $n=10^6$. For the same set of sample sizes, we calculated the average error rate bound derived from equation \eqref{eq:nbound} as $\epsilon_n = \exp(-n\gamma\delta)$ (the average is used as the estimate of $\gamma$ is different across independent runs). Figure \ref{fig:results_psi1}(a) shows $E_n$ and $\epsilon_n$ as a function of $n$ for different values of $\delta$. Importantly, Figure \ref{fig:results_psi1}(a) demonstrates that the theoretical bounds on the error rate are reliable in practice, since the empirical error rate is consistently below this upper bound ($E_n \le \epsilon_n$ for all examined $n,\delta$).

We next look at results for sequential hypothesis testing with indifference region. In all cases we use $\xi:=0.3$ to set the set the sample sizes at which a test is performed (see \eqref{eq:xichoice}). We refer to the number of samples collected in the Markov chain before a decision is made as the \emph{stopping time} (see also Section \ref{sec:stoptime}).
Figures \ref{fig:results_psi1}(b-c) show the mean empirical stopping times for a range of $r$ values for different values of $\delta$ (b), and different values of $\epsilon$ (c). For values of $r$ close to $\E_{\pi}f$, some chains did not stop within $2\cdot 10^6$ samples, and the corresponding mean values are therefore not determined. In both plots, the average of the fixed sample sizes needed for each chain is also shown. 
In Figure \ref{fig:results_psi1} (d), the empirical cumulative distribution of stopping times is shown for the hypothesis test for a set of $r$ values in $(0,1)$. Here the value of $\delta=0.05$ and $\epsilon=0.01$ is fixed. As a reference, we also show the empirical distribution of sample sizes needed to perform the fixed length test. These differ across chains due to the different estimates of the spectral gap. 
The plot shows that for values of $r$ distant from the true average, sequential sampling consistently terminates with small variability at low sample sizes. When $r$ is close to the true average, the stopping times show higher variability. 

Finally, we used the sequential hypothesis test without indifference region. Again, we set $\xi:=0.3$ for all experiments. Figure \ref{fig:results_psi1}(e) shows empirical stopping times for different values of $\epsilon$ and a range of $r$ values. Figure \ref{fig:results_psi1} (f) shows the empirical cumulative distribution of stopping times for the hypothesis test for a set of $r$ values. Here the value of $\epsilon=0.01$ is fixed. 

We evaluated the empirical error rate in the sequential hypothesis test with indifference region. We found that out of all $1000$ runs, under all examined choices of $r,\epsilon,\delta$,  the worst empirical error rate was $3\cdot 10^{-3}$, which was below all choices of $\epsilon$. This shows that the specified error bound of Proposition \ref{propseqtest} was indeed met. It also suggests that the bound might not be sharp and $M$ could be chosen even smaller than described by \eqref{eq:seqtestM}, resulting in earlier stopping. Similarly, for the test without indifference region, the empirical error rate was at most $10^{-3}$ for any of the runs under examined choices of $r$ and $\epsilon$. 

\begin{figure*}[ht]
\centering
	\begin{tabular}{ccr}
	\addtolength{\subfigcapskip}{0.2cm}
	
	\subfigure[Empirical error rates for the fixed sample size test for a range of sample sizes. Dashed lines show the theoretical upper bounds derived from \eqref{eq:nbound}. Here $r = \widehat{E} - \delta$ and $\epsilon=0.01$ are fixed, and $3$ distinct $\delta$ values are shown.]{
		\includegraphics[width=0.4\textwidth]{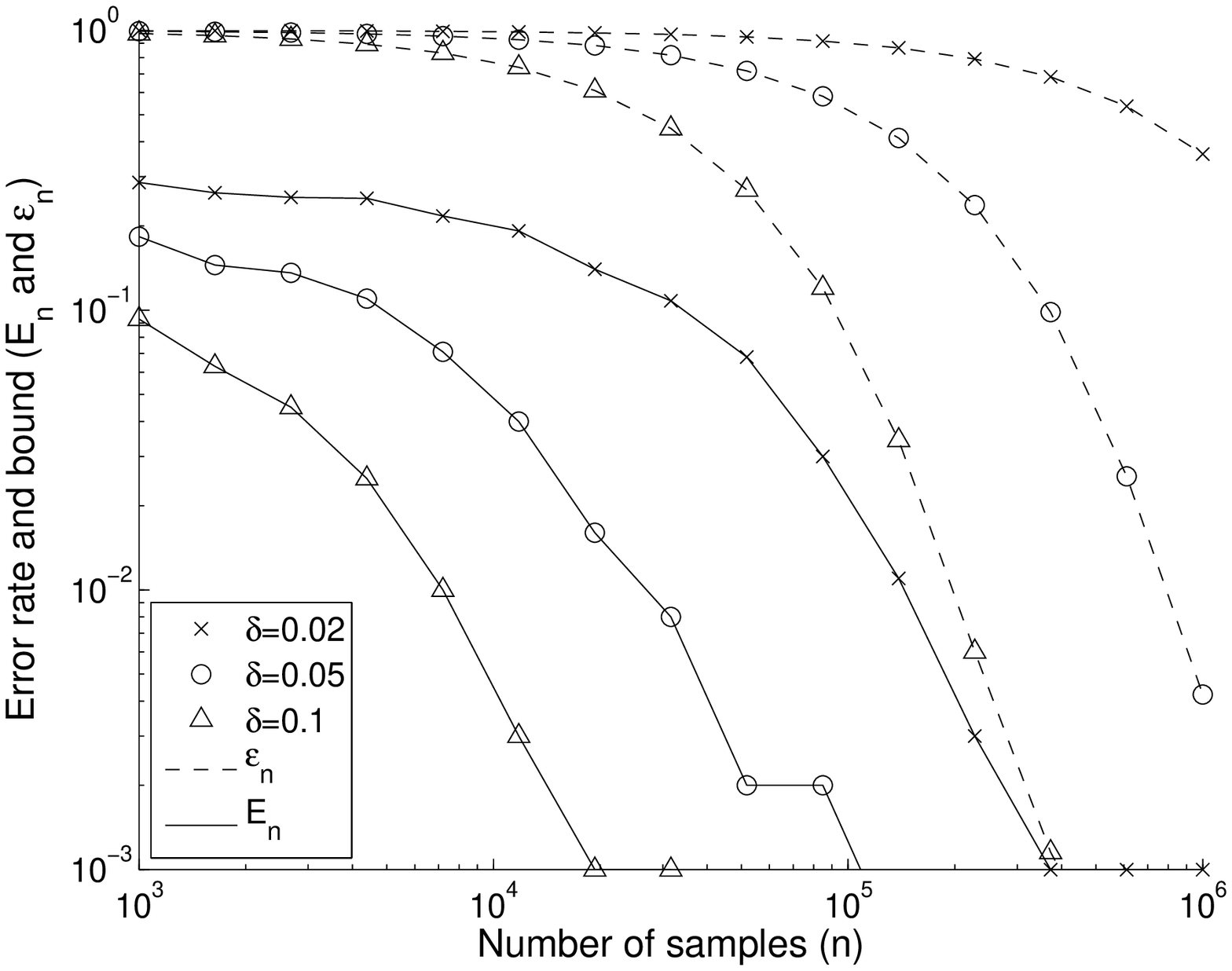}
		}
	
	&
	\addtolength{\subfigcapskip}{0.2cm}
	\subfigure[Average empirical stopping times for sequential hypothesis test with indifference region for different values of $\delta$, with $\epsilon=0.01$. Dashed lines show average sample sizes required for the fixed sample size test.]{
			\includegraphics[width=0.4\textwidth]{figure_seqindetdelta.eps}
			}
	&\\
	\addtolength{\subfigcapskip}{0.2cm}
	\subfigure[Average empirical stopping times for sequential hypothesis test with indifference region for different values of $\epsilon$, with $\delta=0.05$. Dashed lines show average sample sizes required for the fixed sample size test.]{
		\includegraphics[width=0.4\textwidth]{figure_seqindeteps.eps}
		}
	
	&
		\addtolength{\subfigcapskip}{0.2cm}
		\subfigure[Empirical distribution of stopping times with sequential hypothesis test with indifference region for different values of $r$. Here $\delta=0.05$ and $\epsilon=0.01$ is used.]{
			\includegraphics[width=0.4\textwidth]{figure_seqinddist.eps}
			}	
	&\\
	\addtolength{\subfigcapskip}{0.2cm}
		\subfigure[Average empirical stopping times for sequential hypothesis test without indifference region for different values of $\epsilon$.]{
		\includegraphics[width=0.4\textwidth]{figure_seqnoindeteps.eps}
		}	
	&
	\addtolength{\subfigcapskip}{0.2cm}
			\subfigure[Empirical distribution of stopping times with sequential hypothesis test without indifference region for different values of $r$. Here $\epsilon=0.01$ is used.]{
			\includegraphics[width=0.4\textwidth]{figure_seqnoinddist.eps}
			}	
		
	\end{tabular}
	\caption{Empirical results of hypothesis test error rates and stopping times. }
	\label{fig:results_psi1}
\end{figure*}

\makeatletter{}\section*{Conclusion}
In this paper we proposed hypothesis tests on MCMC estimates. These tests are useful in cases where one is not interested in the exact value of the estimate, but only whether it is below or above a certain threshold. We have stated three different hypothesis tests, a fixed length test, a sequential test with indifference region, and a sequential test without indifference region. Our main theoretical contribution is rigorous error bounds for these test, and on their expected running times.
We have illustrated their usage on a case study from the domain of systems biology, using an ODE model of the JAK-STAT biochemical pathway. In our simulations, the sequential tests have performed well, and their running time can be much shorter than the fixed length test, especially when the true expected value is far away from the threshold.

There are some theoretical and practical questions that remain open for further exploration. We have chosen the details of Algorithms \ref{alg:seq} and \ref{alg:seqnoindiff}, in particular, the function $g(i,\epsilon)$, 
in this specific way in order to be able to show error bounds using concentration inequalities. However, we do not claim that these choices are optimal, and our error estimates could possibly be sharpened using different techniques (in the independent case, there is a well developed literature on how to do sequential tests with or without indifference regions in an optimal way, see \cite{laitzeleung1}, \cite{laitzeleung2}).

\section{Acknowledgements}
DP was supported by an MOE Singapore Academic Research Fund Tier 2 Grant ``Approximate Computational Methods for High-Dimensional Systems''.

\bibliographystyle{spbasic}

\clearpage
\makeatletter{}\section{Appendix}
\subsection{Estimating the spectral gap}
In this section, we propose a procedure to estimate the spectral gap of Markov chains. The procedure is motivated by the following lemma. In the statement of the lemma, we are going to use the scalar product \[\left<f,g\right>_{\pi}:=\sum_{x\in \Omega} f(x)g(x)\pi(x),\]
where $\pi$ is the stationary distribution of a finite state Markov chain.
\begin{lemma}\label{rholimlemma}
Suppose that $X_1,X_2,\ldots$ is a finite state Mar\-kov chain with spectral gap $\gamma$, and transition kernel $P$. Suppose that $f:\Omega\to \R$ satisfies that $\E_{\pi}(f)=0$, and that $f$ is not orthogonal to the eigenspace corresponding to the eigenvector of $P$ of the second largest absolute value (with respect to the scalar product $\left<\cdot ,\cdot \right>_{\pi}$). Let
$\rho_\eta(f):=\E_{\pi}(f(X_0)f(X_\eta))$, then
\begin{equation}\lim_{\eta\to \infty}(|\rho_{\eta}(f)|/\Var_{\pi}(f))^{1/\eta}=1-\gamma^*.\end{equation}
\end{lemma}
\begin{remark}
A similar result can be shown for chains with general state spaces, but for notational simplicity we only consider finite state spaces. In practice, the condition of non-orthogonality is almost always satisfied.
\end{remark}
\begin{proof}
We can write the eigenvalues of the operator $P$ as $1=\lambda_1\ge \lambda_2\ge \ldots \ge \lambda_n$, with $n=|\Omega|$, and corresponding right eigenvectors $f_1,\ldots,f_n$, which are orthonormal with respect to the scalar product $\left<\cdot,\cdot\right>_{\pi}$.
It is clear that $f_1=1$, and using the assumption $\E_{\pi}(f)=0$, $f$ can be decomposed as
\[f=\sum_{i=2}^{n}c_i\cdot f_i,\]
for some constants $\{c_i\}_{i=1}^{n}$. Using this decomposition, we have
\begin{align*}
\frac{\E_{\pi}(f(X_0) f(X_{\eta}))}{\Var_{\pi}(f)}=\frac{\left<f,P^{\eta} f\right>_{\pi}}{\left<f,f\right>_{\pi}}=\frac{\sum_{i=2}^{n}c_i^2 \lambda_i^{\eta}}{\sum_{i=2}^{n}c_i^2},
\end{align*}
and using the condition of non-orthogonality, we have
\begin{align*}\lim_{\eta\to\infty}\left|\frac{\sum_{i=2}^{n}c_i^2 \lambda_i^{\eta}}{\sum_{i=2}^{n}c_i^2}\right|^{1/\eta}=\max_{2\le i\le n}|\lambda_i|=1-\gamma^*. \quad \qedhere
\end{align*}
\end{proof}

In practice, we cannot choose $\eta$ to be infinity, in particular, there are some issues related to the variance of the estimator of $\rho_{\eta}(f)$ that need to be taken in to account.
Given burn-in time $t_0$, and $f(X_{t_{0}+1})$, $\ldots, f(X_{t_0+n})$, we define the empirical mean as  \[\ol{f}=\frac{f(X_{t_{0}+1})+\ldots+ f(X_{t_0+n})}{n},\]
and use the estimator 
\[\hat{\rho}_{\eta}(f):=\frac{\sum_{j=t_0+1}^{n+t_0-\eta}(f(X_j)-\ol{f})(f(X_{j+\eta})-\ol{f})}{n-\eta}.\]
The typical range of  dependence of the elements of the Markov chain is of the order of $1/\gamma^*$, so we expect the standard deviation of $\hat{\rho}_{\eta}(f)/\Var_{\pi}(f)$ to be of the order of 
$\frac{1}{\sqrt{n \gamma^*}}$. Now if we want to use the estimator 
\[1-\gamma^*\approx (\rho_{\eta}(f)/\Var_{\pi}(f))^{1/\eta},\]
the standard deviation of $\rho_{\eta}(f)/\Var_{\pi}(f)$ needs to be much smaller than $(1-\gamma^*)^{\eta}$.

Solving the equation $\frac{1}{\sqrt{n \gamma^*}}=(1-\gamma^*)^{\eta}$ leads to $\eta=\frac{\log(n\gamma^*)}{2\log(1/(1-\gamma^*))}$, so we propose the choice 
\begin{equation}\eta=\frac{\log(n\gamma^*)}{4\log(1/(1-\gamma^*))},\end{equation}
for which the standard deviation of $\rho_{\eta}(f)/\Var_{\pi}(f)$ is much smaller than $(1-\gamma^*)^{\eta}$. A slight inconvenience of this choice is that it depends on the unknown parameter $\gamma^*$.  A way to overcome this issue is by computing the value of $\eta$ iteratively.
Based on Lemma \ref{rholimlemma} and the above observation, we propose the following procedure.
\begin{enumerate}
\item Choose some functions $f_1,\ldots, f_m: \Omega\to \R$ that together determine the location in the state space (for example, if $\Omega\subset \R^m$, then $f_1,\ldots, f_m$ can be chosen as the coordinates). 
\item Run some initial amount of simulations $X_1,\ldots, X_{t_0+n}$, and in every step, save the values $f_1(X_i),\ldots, f_m(X_i)$.
\item Compute $\hat{\gamma}^*$ based on $\eta=1$ and $f=f_1, \ldots, f_m$. Denote the minimum of these values by $\gamma^*_{\min}(1)$. Compute $\eta(1):=\frac{\log(n\gamma^*)}{4\log(1/(1-\gamma^*_{\min}(1)))}$.
\item Inductively assume we have already computed $\eta(i)$. Then compute $\gamma^*_{\min}(i+1)$ based on $\eta(i)$. If $\gamma^*_{\min}(i+1)\ge \gamma^*_{\min}(i)$, then stop, and let $\hat{\gamma}^*:=\gamma^*_{\min}(i)$. Otherwise compute $\eta(i+1)$ and repeat this step.
\item If the initial amount of simulations $n$ satisfies that $n>100/\hat{\gamma}^*$, accept the estimate, otherwise choose $n=200/\hat{\gamma}^*$ and restart from Step 2.
\end{enumerate}
The motivation for choosing $f_1,\ldots, f_m$ in this way is that we use all the available information about the Markov chain, and thus the estimator is expected to be more accurate than if we would only use a subset of this information. The motivation for Step 5 is to ensure that we have sufficient initial data for the estimate.

An illustration of this iterative procedure based on our case study (see Section \ref{sec:results}) is shown in Figure \ref{fig:gamma}. Here we use the components of the Markov chain's state (the value of the model parameters $k_1$ to $k_4$) for the estimation. 
Figure \ref{fig:gamma_hist} shows a histogram of the estimated spectral gaps for the $1000$ independent chains used in our case study.

\begin{figure}[h]
\includegraphics[width=0.45\textwidth]{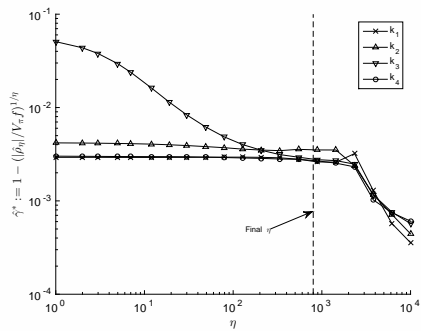}
\caption{Estimated $\hat{\gamma}^*$ for different values of $\eta$ based on the Markov chain states in the space of model parameters $k_1$,$k_2$,$k_3$,$k_4$. The final value of $\eta$ is shown, found according to the proposed iterative procedure.}
\label{fig:gamma}
\end{figure}

\begin{figure}[h]
\includegraphics[width=5cm]{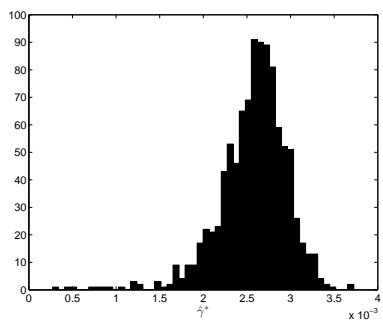}
\caption{Histogram of estimated $\hat{\gamma}^*$ for $1000$ independent chains.}
\label{fig:gamma_hist}
\end{figure}

\subsection{Simulation details}
Here we provide additional details on the JAK-STAT pathway model case study. 
The species in the model are listed in Table \ref{tab:species}. 
\begin{table}[ht]
\centering
\begin{tabular}{|l|p{4.4cm}|p{1cm}|}
\hline
\bf{Name} & \bf{Description} &\bf{Initial amount} \\ \hline
Epo & Erythropoietin, input stimulus & 2.0 \\ \hline
STAT & Unphosphorylated STAT in cytoplasm & 0\\ \hline
STATp & Phosphorylated STAT in cytoplasm & 0\\ \hline
STATpd & Phosphorylated STAT dimer in cytoplasm & 0\\ \hline
STATn & STAT in nucleus & 0\\ \hline
$X_1 \ldots X_K$ & Delay in STAT exiting nucleus & 0\\ \hline
\end{tabular}
\caption{Species in the JAK-STAT pathway model.}
\label{tab:species}
\end{table}

The ODE equations governing the model dynamics are as follows. 
\begin{align*}
	\frac{\mathrm{d}[\mathrm{STAT}]}{\mathrm{d}t} & = -k_1[\mathrm{STAT}][\mathrm{Epo}] + 2k_4[X_K]\\
	\frac{\mathrm{d}[\mathrm{STATp}]}{\mathrm{d}t} & = k_1[\mathrm{STAT}][\mathrm{Epo}] - k_2[\mathrm{STATp}]^2\\
	\frac{\mathrm{d}[\mathrm{STATpd}]}{\mathrm{d}t} & = -k_3[\mathrm{STATpd}] + 0.5k_2[\mathrm{STATp}]^2\\
	\frac{\mathrm{d}[\mathrm{X_1}]}{\mathrm{d}t} & = k_3[\mathrm{STATpd}] -k_4[\mathrm{X}_1]\\
	\frac{\mathrm{d}[\mathrm{X_j}]}{\mathrm{d}t} & = k_4[\mathrm{X}_{j-1}]-k_4[\mathrm{X}_j] \quad , \quad \quad j=2\ldots K \\
	\frac{\mathrm{d}[\mathrm{STATn}]}{\mathrm{d}t} & = k_3[\mathrm{STATpd}]-k_4[\mathrm{X}_K].
\end{align*}
Here $K$ is the number of delay variables, set to $K=10$.

The observation model is defined as follows.
\begin{align*}
y_1 &= [\mathrm{STATp}]+2[\mathrm{STATd}] \\
y_2 &= [\mathrm{STAT}]+[\mathrm{STATp}]+2[\mathrm{STATd}].
\end{align*}
Here $y_1$ represents total phosphorylated STAT and $y_2$ represents total STAT in cytoplasm. Observations are available at $18$ discrete time points up to $60$ minutes. The data points, as well as standard deviations for each point are available in \cite{swameye2003}. 

The parameter vector of the model is $\tht=(k_1,k_2,$ $k_3,k_4)$. We assume a uniform prior over a hypercube defined by a bounded interval for each parameter.
The parameter ranges and the covariance matrix diagonal entries ($\sigma_{\mathrm{MH}}$) used to define the MCMC proposal distribution are provided in Table \ref{table:params}.

\begin{table}[h]
\centering
	\begin{tabular}{c|c|c}
	Parameter 	& Range 	& $\sigma_{\mathrm{MH}}$ \\ \hline 
	$k_1$ 		& $[0,5]$			& $0.02$ \\
	$k_2$		& $[0,30]$			& $0.5$ \\
	$k_3$		& $[0,1]$			& $0.01$ \\
	$k_4$		& $[0,5]$			& $0.02$
	\end{tabular}
\caption{Parameter ranges and entries in the proposal covariance matrix.}
\label{table:params}
\end{table}

\end{document}